\documentclass[10pt,conference]{IEEEtran}
\usepackage{times}

\usepackage{amsmath}
\usepackage{amsfonts}
\usepackage{latexsym}
\usepackage{amssymb}

\usepackage{color}

\usepackage{graphicx}
\usepackage{psfrag}
\usepackage{multirow}
\usepackage{epstopdf}

\usepackage{txfonts}

\usepackage{ntheorem}

\usepackage{enumerate}
\usepackage{accents}

\usepackage{comment}
\theoremstyle{plain}
\theorembodyfont{\normalfont\slshape}

\newtheorem{thm}{Theorem$\!$}
\newenvironment{theorem}
{\begin{thm}\hspace*{-1ex}{\bf.}}{\end{thm}}

\newtheorem{lem}[thm]{Lemma$\!$}
\newenvironment{lemma}{\begin{lem}\hspace*{-1ex}{\bf.}}{\end{lem}}

\newtheorem{prop}[thm]{Proposition$\!$}
\newenvironment{proposition}{\begin{prop}\hspace*{-1ex}{\bf.}}{\end{prop}}

\newtheorem*{prop:nonidea}{Proposition \ref{prop:nonideal_equiv}'}
\newenvironment{prop:nonideal}{\begin{prop:nonidea}\hspace*{-1ex}{\bf.}}{\end{prop:nonidea}}

\newtheorem{cor}[thm]{Corollary$\!$}

\newtheorem{defn}{Definition$\!$}
\newenvironment{definition}{\begin{defn}\hspace*{-1ex}{\bf.}}{\end{defn}}

\newtheorem{xmpl}{Example$\!$}

\newtheorem{alg}{Algorithm$\!$}

\newtheorem{prob}[thm]{Problem$\!$}

\newtheorem{cnst}{Construction$\!$}
\newenvironment{construction}{\begin{cnst}\hspace*{-1ex}{\bf.}}{\end{cnst}}

\newcounter{enumrom}
\renewcommand{\theenumrom}{(\roman{enumrom})}

\makeatletter
\renewcommand{\@endtheorem}{\endtrivlist}
\makeatother

\makeatletter
\renewcommand{\thefigure}{{\@arabic\c@figure}}
\renewcommand{\fnum@figure}{{\bf Figure\,\thefigure}}
\makeatother

\renewcommand{\QEDclosed}{\mbox{\rule[-1pt]{1.3ex}{1.3ex}}} %

\renewcommand{\baselinestretch}{0.98}

\newcommand{\be}[1]{\begin{equation}\label{#1}}
\newcommand{\ee}{\end{equation}}

\renewcommand{\leq}{\leqslant}

\renewcommand{\geq}{\geqslant}

\newcommand{\Cref}[1]{Co\-ro\-lla\-ry\,\ref{#1}}

\newcommand{\code}{\mathcal{C}}
\newcommand{\lcode}{\mathcal{C'}}
\newcommand{\icode}{\mathcal{K}}

\DeclareMathAlphabet{\mathbfsl}{OT1}{cmr}{bx}{it}

\newcommand{\vect}[1]{\boldsymbol{#1}}

\newcommand{\dsub}{\delta}
\newcommand{\dst}{d}
\newcommand{\Dst}{D}
\newcommand{\thalf}{s}

\outer\def\proclaim #1. #2\par{\medbreak
 \noindent{\bf#1.\enspace}{\sl#2\par}%
 \ifdim\lastskip<\medskipamount \removelastskip\penalty55\medskip\fi}

\mathchardef\inn="3232
\renewcommand{\in}{{\,\inn\,}}

\newcommand{\RS}{\overline{RS}}

\begin{document}

\sloppy

%% Paper Title
%% You can use linebreaks \\ within to get better formatting as
%% desired.
\title{Multi-Block Interleaved Codes for Local and Global Read Access}

%% Author names and affiliations:
%%
%% Avoiding spaces at the end of the author lines is not a problem with
%% conference papers because we don't use \thanks or \IEEEmembership.
%%
%% For several authors with only one affiliation:
%%
% \author{
%   \IEEEauthorblockN{Hui-Ting Chang and Stefan M.~Moser}
%   \IEEEauthorblockA{Department of Electrical and Computer Engineering\\
%     National Chiao Tung University (NCTU)\\
%     Hsinchu, Taiwan\\
%     Email: \{email-of-hui-ting,email-of-stefan\}@ieee.org}
% }
%%
%% For up to three affiliations:
%%
%\author{
%\authorblockN{{\bf John Doe}}
%\authorblockA{
%Technion -- Israel Institute of Technology\\
%Department of Electrical Engineering\\
%{\it @ee.technion.ac.il}}
%\and
%\authorblockN{{\bf Jane Doe}}
%\authorblockA{
%Technion -- Israel Institute of Technology\\
%Department of Electrical Engineering\\
%{\it @ee.technion.ac.il}}}

\author{\authorblockN{\textbf{Yuval Cassuto}\authorrefmark{1}, \textbf{Evyatar Hemo}\authorrefmark{1}, \textbf{Sven Puchinger}\authorrefmark{2} and \textbf{Martin Bossert}\authorrefmark{2}}
\authorblockA{\authorrefmark{1}Viterbi Electrical Engineering Department, Technion -- Israel Institute of Technology, Israel \\}
\authorblockA{\authorrefmark{2}Institute of Communications Engineering, Ulm University, Germany\\}
{\it ycassuto@ee.technion.ac.il, evyatar.hemo@gmail.com, \{sven.puchinger,martin.bossert\}@uni-ulm.de}
}

\maketitle

%%
%% For over three affiliations, or if they all won't fit within the width
%% of the page, use this alternative format:
%%
% \author{
%   \IEEEauthorblockN{
%     Michael Shell\IEEEauthorrefmark{1},
%     Homer Simpson\IEEEauthorrefmark{2},
%     James Kirk\IEEEauthorrefmark{3},
%     Montgomery Scott\IEEEauthorrefmark{3} and
%     Eldon Tyrell\IEEEauthorrefmark{4}}
%   \IEEEauthorblockA{
%     \IEEEauthorrefmark{1}School of Electrical and Computer Engineering\\
%     Georgia Institute of Technology, Atlanta, Georgia 30332--0250\\
%     Email: see http://www.michaelshell.org/contact.html}
%   \IEEEauthorblockA{
%     \IEEEauthorrefmark{2}Twentieth Century Fox, Springfield, USA\\
%     Email: homer@thesimpsons.com}
%   \IEEEauthorblockA{
%     \IEEEauthorrefmark{3}Starfleet Academy, San Francisco, California 96678-2391\\
%     Telephone: (800) 555--1212, Fax: (888) 555--1212}
%   \IEEEauthorblockA{
%     \IEEEauthorrefmark{4}Tyrell Inc., 123 Replicant Street, Los Angeles, California 90210--4321}
% }

%% Use for special paper notices

%% To balance the two columns, you should reduce the text-height of
%% the last page using the following command:
%%%%%%%%%%%%%%%%%%%%%%%%%%%%%%%%%%%%%%%%%%%%%%%%%%%%%%%%%%%%%%%%%%%%%
%\addtolength{\textheight}{-9.35cm}
%%%%%%%%%%%%%%%%%%%%%%%%%%%%%%%%%%%%%%%%%%%%%%%%%%%%%%%%%%%%%%%%%%%%%
%% with an appropriate value. This command must be place on the second
%% last page, i.e., for a one-page abstract here, for a two-page
%% abstract right after the \maketitle command.

%% Create the title:
\maketitle
\bibliographystyle{IEEEtranS}
%% Abstract:
%% For the final version of the accepted paper, please make sure you
%% remove the comment "THIS PAPER IS ELIGIBLE FOR THE STUDENT PAPER
%% AWARD."
%%
\begin{abstract}
We define multi-block interleaved codes as codes that allow reading information from either a small sub-block or from a larger full block. The former offers faster access, while the latter provides better reliability. We specify the correction capability of the sub-block code through its gap $t$ from optimal minimum distance, and look to have full-block minimum distance that grows with the parameter $t$. We construct two families of such codes when the number of sub-blocks is $3$. The codes match the distance properties of known integrated-interleaving codes, but with the added feature of mapping the same number of information symbols to each sub-block. As such, they are the first codes that provide read access in multiple size granularities and correction capabilities.
\end{abstract}
\section{Introduction}
The two central features sought in data-storage applications are {\em extreme data reliability} and {\em fast data access}. In data reliability we wish to avoid data loss in all conceivable circumstances, and in fast data access we want to read data with high throughput and low latency. Access considerations often dictate implementing an error-correcting code over a fixed (and not too large) data unit, which degrades the coding performance and the resulting reliability. An especially interesting instance of this happens in Flash storage, where a group of multi-level memory cells is divided to smaller data units called {\em pages}.

It is well known that coding over large block lengths gives the best reliability for a given coding rate. But adding access performance to the considerations, block lengths are tightly constrained by the granularity required for the read/write interface of the storage device. The standard approach of the storage industry is to fix the coding block length to be the basic access unit of the device (e.g. a page of a few KB), and come up with the best code for this block length. A conflict between coding efficiency and access performance thus emerges whenever the storage is required to deliver data at fine access granularities that imply short code blocks. Solving this conflict needs a multi-level read/write scheme with the following features:
\begin{enumerate}
\item Jointly writing $m$ data (sub-)units (e.g. pages), each with $k$ symbols, to a single write block.
\item Allowing random read access to a sub-unit, i.e., returning its $k$ symbols by physically reading a sub-block $1/m$ smaller than the write block.
\item In cases where a sub-unit read fails due to excess errors in the sub-block, reading the full write block succeeds in retrieving the data.
\end{enumerate}
Feature 2 takes care of the fast-access requirement, and feature 3 improves reliability with only a minor adjustment of read performance in rare unfortunate instances. Note that focusing on random sub-unit {\em read} performance is in line with real applications that are much more sensitive to delayed reads than writes.

The formal model that fits the above three features now follows. Define a {\em write block} as $N=mn$ symbols, where $m$ is the number of {\em sub-blocks} and $n$ is the number of symbols in each sub-block. A {\em write unit} of $K=mk$ information symbols divides into $m$ {\em sub-units} of size $k$ each. In the write path, $K$ information symbols are encoded into $N$ code symbols and written to the memory. In the read path, we need to return a sub-unit of $k$ symbols by reading a sub-block of $n$ symbols. This operation is divided into {\em decoding} and {\em reverse mapping}, where the former corrects the errors/erasures in the $n$ sub-block symbols, and the latter retrieves the $k$ information symbols from the corrected sub-block. If sub-block decoding failed to correct the errors/erasures, we allow reading the entire write block and decode it as $[N,K]$ code.
%In most decoding instances the redundancy in the $n$ sub-block code symbols will suffice to recover the $k$ information symbols; while in instances with too many errors in the sub-block we will decode a larger block with more redundancy, and likely a lower error rate.
We define an $\langle N,K,m\rangle$ {\em multi-block interleaved code} as a code supporting these operations with the parameters $N,K,m$. Our objective in this paper is to construct $\langle N,K,m\rangle$ codes with good distance properties for both individual sub-block decoding and full write-block decoding.

The first place to look for multi-block interleaved codes is within the literature on {\em integrated-interleaving (I-I) codes}~\cite{IImag,Subline} (see also~\cite{MarioII}), which provide one minimum distance for sub-block codewords and a larger one for full codewords. While I-I codes give good (sometimes provably optimal) tradeoffs between sub-block and full-codeword distances, they cannot be readily used as multi-block interleaved codes. The issue with I-I codes is that it is not known how to reverse-map a sub-block codeword to $k$ information symbols. The encoder presented for I-I codes~\cite{MarioII} maps different numbers of information symbols to different sub-blocks, and it is not clear how to re-organize the resulting $mk\times mn$ generator matrix to get $m$ full-rank $k\times n$ sub-matrices. Because of that, I-I codes fail to support the critical feature 2 in the read/write model. Another related coding framework is {\em locally recoverable codes (LRC)}~\cite{GopalanLRC,KumarVLocalErasures,TamoBargLRC}. LRC are focused on efficient local {\em repair} of individual code symbols, while here we are interested in local decoding of larger sub-blocks. The primary reason why known LRCs are not directly applicable to the proposed model is that they partition the code coordinates to disjoint sets that are each an MDS local code. Making each of the $m$ sub-blocks an $[n,k]$ MDS local code degenerates the problem, because all parity symbols are local in this case.

Our main contributions in the paper are two code constructions of multi-block interleaved codes providing flexible tradeoffs between sub-block and full-codeword correction capabilities. The flexibility is achieved by introducing an integer parameter $t$ specifying the gap of the sub-block code from minimum-distance optimality. As $t$ grows, the constructed codes enjoy better minimum distances for the full write-block codewords. This new view of the local vs. global correction tradeoff through the $t$ parameter allows their explicit joint-optimization, while in I-I codes this relation is much less transparent (in I-I codes there is no notion of sub-block dimension $k$, so $t$ is not even defined). One family of codes, presented in Section~\ref{sec:m3}, allows improvement in the full-block minimum distance (by growing $t$) reaching $d=1.5(n-k+1)$. An improved construction, presented in Section~\ref{sec:im3}, can grow the full-block minimum distance further until $d=1.6(n-k+1)$. The proposed constructions are given for $m=3$ sub-blocks, but can be extended to additional $m$ with similar techniques. The key idea toward having sub-block reverse mapping is the specification of the codes through their generator matrices. This enables to carefully populate the generator matrix with constituent Reed-Solomon (RS) generators such that all distance properties are fulfilled simultaneously. Coincidentally, both constructions achieve the same distance properties as known I-I codes: the former meets 2-level I-I and the latter meets 3-level I-I. However, the constructed codes are {\em not} reformulations of these known I-I codes: the resulting codes are different, and there is no apparent transformation from the I-I codes to our codes that adds the desired sub-block access features. A previous use of constituent RS generators was in~\cite{dettmar1993new} to construct {\em partial unit memory (PUM)} convolutional codes~\cite{lee1976short,lauer1979some} (see also~\cite{MartinBook:99}).

\section{Definitions and Notations}
Throughout the paper we denote the set $\{a,a+1,\ldots,b\}$ by $[a:b]$. Similarly, $[a:b]_{n}$ denotes the same set, but with element indices taken modulo $n$. For example, if $n=7$, then $[5:9]_{n}=[5:2]_{n}=\{5,6,0,1,2\}$. The operation $A\setminus B$ represents set difference between $A$ and $B$. We use the term {\em distance} to refer to the Hamming distance between vectors over the field GF$(q)$. For notational convenience, the statements in the paper on code minimum distances are written as some $d$ equal some number, even though in some cases we only prove (the important part) that $d$ is {\em at least} that number. We first define our principal object of study, which we call {\em multi-block interleaved code}.
\begin{definition}[multi-block interleaved code]
An $[N,K]$ linear code $\code$ is a $\langle N,K,m\rangle$ \textbf{multi-block interleaved code} if its code coordinates are partitioned to $m$ disjoint sub-blocks of size $n=N/m$ each, and from each sub-block we can recover a disjoint size-$k=K/m$ sub-unit of information symbols.
\end{definition}

For decoding a multi-block interleaved code we assume $m$ projections of the coordinates $[1:N]$ to disjoint subsets $[1:n]$, $[n+1:2n]$,... up to $[(m-1)n+1:mn]$. We assume that the projection of $\code$ onto any of these subsets gives the same $[n,\ell]$ code, which we denote $\lcode$. Next we define the distances of multi-block interleaved codes.
\begin{definition}[sub-block minimum distance]
An $\langle N,K,m\rangle$ multi-block code $\code$ has \textbf{sub-block minimum distance} $\dsub$ if the sub-block projected code $\lcode$ has minimum distance $\dsub$.
\end{definition}
The (full-block) minimum distance of an $\langle N,K,m\rangle$ multi-block code will be defined to be the usual minimum distance as $[N,K]$ block code.
\subsection{Bound on the minimum distance of multi-block interleaved codes}
The following is a restatement of a known theorem from~\cite{Subline}, but posed in different notation (using the $t$ parameter) that is helpful for the constructions in the sequel.
\begin{theorem}\label{th:dist_bound}
Let $\code$ be an $\langle N=mn,K=mk,m\rangle$ multi-block interleaved code with sub-block minimum distance $\dsub=n-k-t+1$. If $t\leq(k-1)/(m-1)$ then the minimum distance of $\code$ is bounded by
\begin{equation} \dst \leq n-k+(m-1)t+1.\label{eq:d_bound}\end{equation}
\end{theorem}
\begin{comment}
\begin{proof}
We use the well-known fact that an $s$-subset of columns of the generator matrix for $\code$ with rank at most $K-1$ implies minimum distance of at most $N-s$. The same fact was used to prove bounds on minimum distance of locally recoverable codes (LRC)~\cite{GopalanLRC,KumarVLocalErasures,TamoBargLRC}. For the subset we collect $(m-1)n$ columns corresponding to $m-1$ full sub-blocks, plus $k-(m-1)t-1\geq 0$ columns from the $m$-th sub-block. Knowing the sub-block distance, the total rank of the column subset is at most $(m-1)(k+t)+k-(m-1)t-1=K-1$; the cardinality of the subset is $s=(m-1)n+k-(m-1)t-1$, which gives the bound by taking $N-s$.\hfill
\end{proof}
\end{comment}
An optimal multi-block interleaved code is one meeting the bound~\eqref{eq:d_bound} with equality. We note two interesting special cases of Theorem~\ref{th:dist_bound}. For $m=1$ we get the usual Singleton bound; for $m=2$ we get $\dst \leq n-k+t+1$ when $t\leq k-1$. The upper bound reveals the tradeoff between the sub-block minimum distance and the minimum distance of the code. The parameter $t$ prescribes the gap of the sub-block distance from optimality, and in return increasing $t$ improves the upper bound on the distance of the full-block code.
\subsection{Short review of Reed-Solomon codes}
The key building block for our constructions is Reed-Solomon (RS) codes, which are now defined.
\begin{definition}\label{def:RS}[Reed-Solomon code]
Let $\alpha$ be an element of a finite field GF$(q)$ of order $n$. Then define a $[n,\ell]$ \textbf{Reed-Solomon code} as all polynomials $c(x)$ obtained as
\[ c(x) = i(x)g(x),\]
where $i(x)$ is an arbitrary information polynomial of degree $\leq \ell-1$, and
\begin{equation} g(x) = (x-\alpha^{-s})(x-\alpha^{-(s+1)})\cdots(x-\alpha^{-(s+n-\ell-1)}),\label{eq:RS_gen}\end{equation}
and $s$ is some integer in $[0:n-1]$.
\end{definition}
In simple words, an RS code is obtained by a generator polynomial $g(x)$ with $n-\ell$ roots whose reciprocals are consecutive powers of $\alpha$. As a useful compact notation, given $n$ and $q$, we define a RS code through its (cyclically) consecutive power set. In this representation, the code defined by the generator in~\eqref{eq:RS_gen} is
\[RS([s:s+r-1]_{n}),\]
where $r=n-\ell$ is called the {\em redundancy} of the code. It is well known that the minimum distance of $RS([s:s+r-1]_{n})$ is $r+1$. It will often be convenient to represent an RS code by the {\em complement} of its root power set, that is,
\[\RS([a:b]_{n})\triangleq RS([0:n-1]_{n}\setminus[a:b]_{n}).\]
The {\em dimension} of the code $\RS([a:a+\ell-1]_{n})$ is $\ell$, and its minimum distance is $n-\ell+1$.

\section{Construction of Multi-Block Interleaved Codes}\label{sec:m3}
Our focus in the paper is on codes with $m=3$, because this is the most useful case for Flash-based storage devices with $8=2^3$ representation levels. The constructions can be generalized to larger $m$ values.
For the specific case of $m=3$ we add the following definitions.
\begin{definition}\label{def:D_i}[$\Dst_i$-minimum-distances]
For $i\in\{1,2,3\}$, an $\langle N,K,3\rangle$ multi-block interleaved code $\code$ has \textbf{$\Dst_i$-minimum-distance} $\dst_i$ if the lowest weight of any codeword with exactly $i$ non-zero sub-blocks is $\dst_i$.
\end{definition}
Note that according to Definition~\ref{def:D_i}, the minimum distance of a $\langle N,K,3\rangle$ code is $\min_{i\in\{1,2,3\}}\dst_i$. In addition to serving as upper bounds on the code minimum distance, the $\Dst_i$-minimum-distances have important operational meanings: the $\Dst_1$-minimum-distance gives the correction capability when $1$ sub-block suffers many errors/erasures, and the $\Dst_2$-minimum-distance does the same when $2$ sub-blocks suffer many errors/erasures.
\begin{construction}\label{cnst:m3}
Let $\code_{3}$ be defined by a generator matrix $G=\left[\begin{array}{c}G_1 \\ G_2 \\ G_3 \end{array}\right]$, where $G_1$, $G_2$ and $G_3$ are the $k\times 3n$ matrices given in Fig.~\ref{fig:G_m3}. Now we define the component matrices of Fig.~\ref{fig:G_m3}; blank rectangles represent all-zero sub-matrices.
\begin{itemize}
\item $G1$ is a generator matrix for the code $\RS([0:t-1]_{n})$.
\item $G2$ is a generator matrix for the code $\RS([t:2t-1]_{n})$.
\item $GI$ is a generator matrix for the code $\RS([2t:k-1]_{n})$.
\item $GE$ is a generator matrix for the code $\RS([k:k+t-1]_{n})$.
\end{itemize}
\end{construction}
\begin{figure}[h]
\psfrag{G0}{$GI$}
\psfrag{G1}{$G1$}
\psfrag{G2}{$G2$}
\psfrag{G}{$GE$}
\psfrag{k1}{$k-2t$}
\psfrag{t}{$t$}
\psfrag{n}{$n$}
\psfrag{G1e}{$G_{1}=$}
\psfrag{G2e}{$G_{2}=$}
\psfrag{G3e}{$G_{3}=$}
    \center
    \includegraphics[scale=0.4]{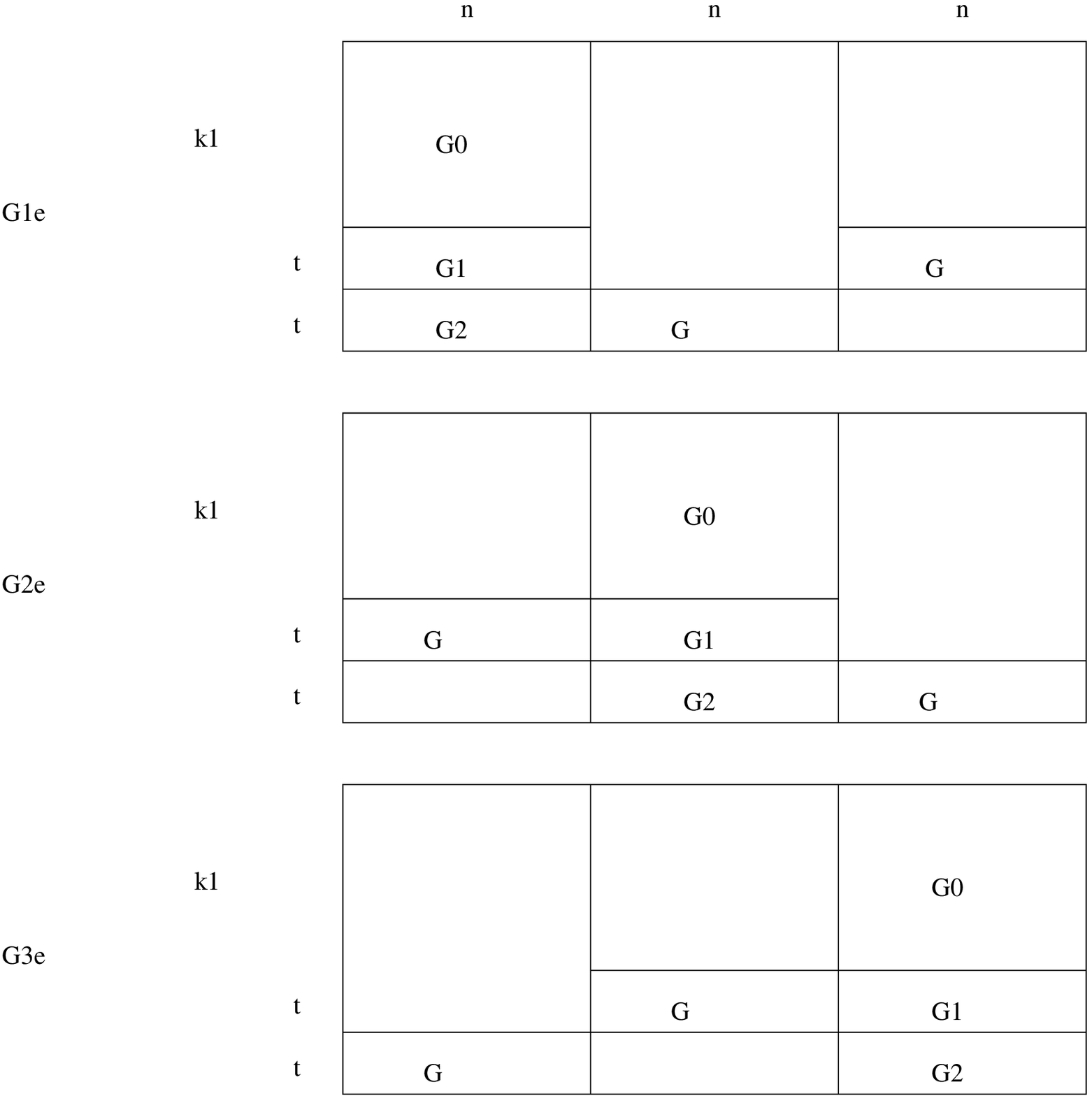}
    \caption{\small{Generator matrices for a $\langle 3n,3k,3\rangle$ multi-block interleaved code $\code_3$.}}\vspace{-1ex}
    \label{fig:G_m3}
\end{figure}

First note that the dimensions of the codes specified for $G1,G2,GI,GE$ fit their corresponding sub-matrices in Fig.~\ref{fig:G_m3}. It is straightforward to write a generator matrix for any code of the type $\RS([a:a+\ell-1]_{n})$, because the codewords of this code are evaluations of polynomials of the form $x^{a}U(x)$, where $U(x)$ is a polynomial of degree $\ell-1$ or less.
 As required, the dimension of $\code_3$ is $3k$, and its length is $3n$. We defer the discussion on encoding and decoding till after we prove the minimum distances in the following.
\begin{proposition}\label{prop:cnst_m3_subdist}
When $t<k/2$, $\code_3$ has sub-block minimum distance $\dsub=n-k-t+1$.
\end{proposition}
\begin{proof}
Looking at any length-$n$ sub-block, we observe from the structure of $G$ that the codeword projected on the sub-block is a linear combination of codewords generated by $G1,G2,GI,GE$. Noting the root power sets of these codes (the complements of the sets listed in Construction~\ref{cnst:m3}), we get: $G1\rightarrow RS([t:n-1]_{n})$, $G2\rightarrow RS([2t:t-1]_{n})$, $GI\rightarrow RS([k:2t-1]_{n})$, and $GE\rightarrow RS([k+t:k-1]_{n})$. If a power index lies in the intersection of these sets, then every polynomial of a codeword in the sub-block code has this power index as root. Taking the intersection of these sets we obtain that the projected sub-block code is a
\[RS([k+t:n-1]_{n}) \]
code, and thus its minimum distance is $n-k-t+1$.\hfill
 \end{proof}
Given the sub-block minimum distance of Proposition~\ref{prop:cnst_m3_subdist}, the next lemma shows that the $\Dst_1$-minimum-distance of Construction~\ref{cnst:m3} is optimal with respect to the bound\footnote{The bound is given on the minimum distance, but from the proof of Theorem~\ref{th:dist_bound} it readily applies to the $\Dst_1$-minimum-distance as well.} of Theorem~\ref{th:dist_bound}.
\begin{lemma}\label{lem:cnst_m3_dist_1}
A code $\code_3$ from Construction~\ref{cnst:m3} has $\Dst_1$-minimum-distance $\dst_1=n-k+2t+1$, for all $t<k/2$.
\end{lemma}
\begin{proof}
Denote the length $3k$ information vector as $\vect{v}$, and partition it as
\[ \vect{v}=(\vect{v}_{1},\vect{v}_{2},\vect{v}_{3}).\]
Each sub-vector $\vect{v}_{j}$ has length $k$, and we further partition it to
\[ \vect{v}_{j} = (\vect{v}_{j,I},\vect{v}_{j,1},\vect{v}_{j,2}).\]
The lengths of the $m=3$ constituent vectors are from left to right: $k-2t,t,t$.  Consider a codeword of $\code_3$ with one non-zero sub-block. Assume (wlog due to symmetry) that the left sub-block is non-zero. It is clear from the structure of $G$ in Fig.~\ref{fig:G_m3} that all sub-vectors except $\vect{v}_{1,I}$ must be all-zero. In this case, the codeword weight is at least the minimum weight of a non-zero codeword from $\RS([2t:k-1]_{n})$ (spanned by $GI$), which is $n-k+2t+1$ as required by the lemma statement.\hfill
 \end{proof}
The behavior of the $\Dst_2$-minimum-distance as a function of $t$ is quite different, as we now see.
\begin{lemma}\label{lem:cnst_m3_dist_2}
A code $\code_3$ from Construction~\ref{cnst:m3} has $\Dst_2$-minimum-distance $\dst_2=2(n-k-t+1)$.
\end{lemma}
The proof of Lemma~\ref{lem:cnst_m3_dist_2} is trivial from the fact that $\dsub=n-k-t+1$. Unfortunately, the construction does not exclude codewords with two minimum-weight sub-block codewords. The problem lies in the fact that both $\vect{v}_{1,1}$ and $\vect{v}_{2,2}$ are multiplied by the same matrix $GE$ in the right column, which may cancel their contribution to the right sub-block. We now summarize the distance properties   in the following theorem (proof immediate).
\begin{theorem}\label{th:m3_dists}
A code $\code_3$ from Construction~\ref{cnst:m3} has sub-block minimum distance $\dsub=n-k-t+1$, $\Dst_1$-minimum-distance $\dst_1=n-k+2t+1$, and $\Dst_2$-minimum-distance $\dst_2=2(n-k-t+1)$. For $t\leq (n-k+1)/4$, $\code_3$ has minimum distance $\dst=\dst_1$ (optimal), and for $(n-k+1)/4<t$ it has minimum distance $\dst=\dst_2$ (not optimal).
\end{theorem}
Similar to the case here, all the known I-I constructions in the literature with optimal $\Dst_1$-minimum-distance (given $t$) suffer from the same problem of having codewords with two minimum-weight sub-block codewords. In~\cite{Subline} the authors address this exact problem, but give no explicit constructions. The implication of Theorem~\ref{th:m3_dists}, in particular the terms $2t$ in $\dst_1$ and $-2t$ in $\dst_2$, is that there is a ratio of $1$ between the loss in $\dst_2$ to the gain in $\dst_1$ as we increase $t$. This high ratio compromises the code's ability to address high error/erasure events in both one and two sub-blocks. In our next construction we solve this by proving a much lower ratio of $2/3$. But before that, we want to show the good property of Construction~\ref{cnst:m3} having an extremely efficient reverse mapping from the sub-block codeword to the $k$ information symbols of the input sub-unit $\vect{v}_j$. We now write the encoding and reverse-mapping functions explicitly.\\
\textbf{Encoding:} The encoder input is an information vector $\vect{v}=(\vect{v}_1,\vect{v}_2,\vect{v}_3)$, where each $\vect{v}_j$ is a vector of $k$ GF$(q)$ symbols holding the information sub-unit for sub-block $j$. The output of the encoder is simply the vector $\vect{v}\cdot G$, which is a codeword of $\code_3$. The matrices $G1,G2,GI,GE$ are chosen as the standard polynomial-evaluation matrices of the respective RS codes (mapping $\ell$ polynomial coefficients to $n$ code symbols). For a code $\RS([a:a+\ell-1]_{n}$, the corresponding generator matrix maps an information sub-vector $\vect{u}$ to a codeword by evaluating the polynomial $\tilde{U}(x)=x^aU(x)$ on $n$ elements of GF$(q)$, where $U(x)$ is the polynomial whose coefficients are the elements of $\vect{u}$, and it has degree at most $\ell-1$. Note that this is not the way the RS code was defined in Definition~\ref{def:RS}, but it is true by an equivalent definition of RS codes.\\
\textbf{Reverse mapping:} The reverse-mapping operation returns the information sub-unit $\vect{v}_j$ from the codeword of sub-block $j$. Denote by $u(x)$ the polynomial whose coefficients are the codeword symbols corresponding to $\vect{u}$. From the property given in Definition~\ref{def:RS}, $u(\alpha^{-i})$ is zero for all $i$ except those in $[a:a+\ell-1]_{n}$. Since these sets are not overlapping between $G1,G2,GI,GE$, we can take the combined sub-block codeword polynomial $c(x)$ (which is a sum of different polynomials like $u(x)$ above), and get that for each $i\in[0:k-1]$, that $c(\alpha^{-i})=u(\alpha^{-i})$ for one of the sub-vectors $\vect{u}$ in $\vect{v}_{j}$. Now we can use the well known property of RS codes (inverse DFT) and obtain
\[ U_l = n^{-1} u(\alpha^{-(a+l)}) = n^{-1} c(\alpha^{-(a+l)}),~ l\in[0:\ell-1].  \]
Since $c(x)$ is given, we can find the coefficients of all the polynomials $U(x)$ and thus all the vectors $\vect{u}$.\\
Neither the I-I codes from~\cite{IImag} nor the ones in~\cite{Subline} have a known way to map $k$ information symbols to each sub-block codeword, such that it is possible to retrieve these $k$ symbols back from the sub-block codeword. The encoder specified for I-I codes~\cite{MarioII} (corresponding to the parameters of  Construction~\ref{cnst:m3}) maps $k+t$ information symbols to each of two sub-blocks, and $k-2t$ to the third.
\section{An Improved Construction with Larger Minimum Distance}\label{sec:im3}
To get codes $\langle N,K,m\rangle$ with larger minimum distances beyond what Construction~\ref{cnst:m3} can achieve, we present the following improved construction. In the sequel we define $\thalf=t/2$, assuming even $t$.
\begin{construction}\label{cnst:im3}
Let $\icode_{3}$ be defined by a generator matrix $G=\left[\begin{array}{c}G_1 \\ G_2 \\ G_3 \end{array}\right]$, where $G_1$, $G_2$ and $G_3$ are $k\times 3n$ matrices specified in Fig.~\ref{fig:G_im3} as follows. The figure specifies $4\thalf$ rows of each $G_i$, which are populated by the generator matrices specified below, or combinations thereof. The figure omits the top $k-4\thalf$ rows of each $G_i$ that host the $GI$ matrix given below, in the corresponding sub-block \footnote{same as in Fig.~\ref{fig:G_m3}.}.
\begin{itemize}
\item $G1$ is a generator matrix for the code $\RS([0:\thalf-1]_{n})$.
\item $G2$ is a generator matrix for the code $\RS([\thalf:2\thalf-1]_{n})$.
\item $G3$ is a generator matrix for the code $\RS([2\thalf:3\thalf-1]_{n})$.
\item $G4$ is a generator matrix for the code $\RS([3\thalf:4\thalf-1]_{n})$.
\item $GI$ is a generator matrix for the code $\RS([4\thalf:k-1]_{n})$.
\item $GE$ is a generator matrix for the code $\RS([k:k+\thalf-1]_{n})$.
\item $GF$ is a generator matrix for the code $\RS([k+\thalf:k+2\thalf-1]_{n})$.
\end{itemize}
\end{construction}
\begin{figure}[h]
\psfrag{G1}{$G1$}
\psfrag{G2}{$G2$}
\psfrag{G3}{$G3$}
\psfrag{G43}{$G4+G3$}
\psfrag{E}{$GE$}
\psfrag{F}{$GF$}

\psfrag{s}{$4\thalf$}
\psfrag{v11}{$\vect{v}_{1,1}:$}
\psfrag{v12}{$\vect{v}_{1,2}:$}
\psfrag{v13}{$\vect{v}_{1,3}:$}
\psfrag{v14}{$\vect{v}_{1,4}:$}
\psfrag{v21}{$\vect{v}_{2,1}:$}
\psfrag{v22}{$\vect{v}_{2,2}:$}
\psfrag{v23}{$\vect{v}_{2,3}:$}
\psfrag{v24}{$\vect{v}_{2,4}:$}
\psfrag{v31}{$\vect{v}_{3,1}:$}
\psfrag{v32}{$\vect{v}_{3,2}:$}
\psfrag{v33}{$\vect{v}_{3,3}:$}
\psfrag{v34}{$\vect{v}_{3,4}:$}
\psfrag{n}{$n$}
\psfrag{G1e}{$G_{1}=$}
\psfrag{G2e}{$G_{2}=$}
\psfrag{G3e}{$G_{3}=$}
    \center
    \includegraphics[scale=0.4]{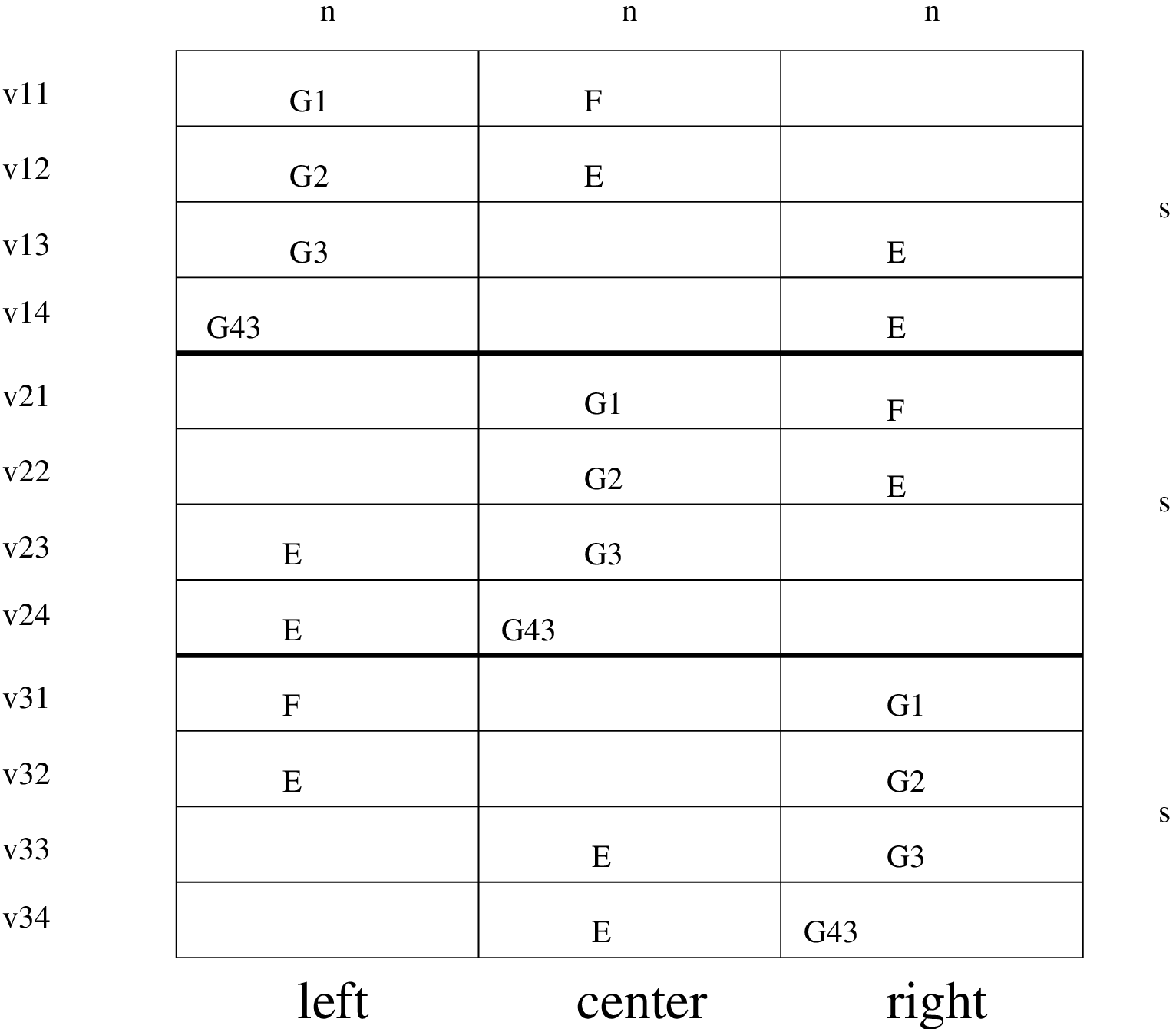}
    \caption{\small{Generator matrices for an improved $\langle 3n,3k,3\rangle$ multi-block interleaved code $\icode_3$.}}\vspace{-1ex}
    \label{fig:G_im3}
\end{figure}

As in the previous construction the dimensions of the codes specified for $G1,G2,G3,G4,GI,GE,GF$ fit their corresponding sub-matrices in Fig.~\ref{fig:G_im3}. The code dimension of $\icode_3$ is $3k$, and its length is $3n$. We prove the minimum distances in the following.
\begin{proposition}\label{prop:icnst_m3_subdist}
When $t<k/2$, $\icode_3$ has sub-block minimum distance $\dsub=n-k-t+1$.
\end{proposition}
\begin{proof}
The proof is identical to Proposition~\ref{prop:cnst_m3_subdist}: at each sub-block there is a codeword of an RS code with $n-k-2\thalf=n-k-t$ roots with consecutive powers.\hfill
 \end{proof}
\begin{lemma}\label{lem:icnst_m3_dist_1}
The code $\icode_3$ has $\Dst_1$-minimum-distance $\dst_1=n-k+3t/2+1$.
\end{lemma}
\begin{proof}
Recall the partition $\vect{v}=(\vect{v}_{1},\vect{v}_{2},\vect{v}_{3})$.
Now we partition each $\vect{v}_{j}$ to
\[ \vect{v}_{j} = (\vect{v}_{j,I},\vect{v}_{j,1},\vect{v}_{j,2},\vect{v}_{j,3},\vect{v}_{j,4}).\]
The lengths of the constituent sub-vectors are from left to right: $k-4\thalf,\thalf,\thalf,\thalf,\thalf$. Consider a codeword of $\icode_3$ with one non-zero sub-block, for which we seek a lower bound on the Hamming weight. Assume (wlog due to symmetry) that the one non-zero sub-block is the left one. Then it is clear from Fig.~\ref{fig:G_im3} that the only sub-vectors that can be non-zero are $\vect{v}_{1,I}$, $\vect{v}_{1,3}$, and $\vect{v}_{1,4}$. Setting any other sub-vector to non-zero implies a non-zero codeword in at least one other sub-block. Further, if $\vect{v}_{1,3}$ or $\vect{v}_{1,4}$ are non-zero, then they must satisfy $\vect{v}_{1,3}+\vect{v}_{1,4}=0$. Otherwise the codeword in the right sub-block will be non-zero as well. It can be seen that this last condition implies that only $G4$ contributes to the non-zero codeword, and together with the contribution of $GI$ from  $\vect{v}_{1,I}$, we get $n-k+3\thalf$ consecutive roots and $\dst_1=n-k+3t/2+1$.\hfill
\end{proof}

\begin{lemma}\label{lem:icnst_m3_dist_2}
The code $\icode_3$ has $\Dst_2$-minimum-distance $\dst_2=2(n-k-t/2+1)$.
\end{lemma}
\begin{proof}
We use the definition of $\vect{v}$ from the proof of Lemma~\ref{lem:icnst_m3_dist_1}.
Consider a codeword of $\icode_3$ with two non-zero sub-blocks, for which we seek a lower bound on the total Hamming weight. Assume (wlog due to symmetry) that the two non-zero sub-blocks are the left and center ones. Then it is clear from Fig.~\ref{fig:G_im3} that the only sub-vectors that can be non-zero are $\vect{v}_{1,I}$, $\vect{v}_{1,1}$, $\vect{v}_{1,2}$, $\vect{v}_{1,3}$, $\vect{v}_{1,4}$ from $\vect{v}_{1}$, and $\vect{v}_{2,I}$, $\vect{v}_{2,2}$, $\vect{v}_{2,3}$, $\vect{v}_{2,4}$ from $\vect{v}_{2}$. Other non-zero sub-vectors would imply a non-zero codeword on the right sub-block. The above restriction on the sub-vectors that can be non-zero implies that there is no contribution of $F$ in the left sub-block codeword and no contribution of $G1$ in the center sub-block codeword. So each of the sub-block codewords has $n-k-\thalf$ consecutive-power roots, and thus $\dst_2=2(n-k-\thalf+1)$.\hfill
\end{proof}
Combining Proposition~\ref{prop:icnst_m3_subdist} and Lemmas~\ref{lem:icnst_m3_dist_1},\ref{lem:icnst_m3_dist_2}, we summarize the distance properties of Construction~\ref{cnst:im3} in the following theorem (proof omitted).
\begin{theorem}\label{th:im3_dists}
A code $\icode_3$ from Construction~\ref{cnst:im3} has sub-block minimum distance $\dsub=n-k-t+1$, $\Dst_1$-minimum-distance $\dst_1=n-k+3t/2+1$, and $\Dst_2$-minimum-distance $\dst_2=2(n-k-t/2+1)$. For $t\leq 2(n-k+1)/5$, $\icode_3$ has minimum distance $\dst=\dst_1$, and for $2(n-k+1)/5<t\leq (n-k+1)/2$ it has minimum distance $\dst=\dst_2$.
\end{theorem}
Theorem~\ref{th:im3_dists} demonstrates the value of Construction~\ref{cnst:im3}: it allows to use large values of $t$ that give minimum distances higher than Construction~\ref{cnst:m3} can reach. It can be calculated that when $t>2(n-k+1)/7$, Construction~\ref{cnst:im3} has superior minimum distance compared to Construction~\ref{cnst:m3}. From $t=2(n-k+1)/7$ we can increase $t$ in Construction~\ref{cnst:im3} up to $t=2(n-k+1)/5$, and get minimum distance as high as $d=1.6(n-k+1)$, while Construction~\ref{cnst:m3} has a cutoff at $d=1.5(n-k+1)$.\\
\textbf{Erasure decoding:} We chose to prove the code correction properties by lower bounding the minimum distances. However, an alternative way is by giving an explicit erasure-decoding algorithm, which we do now. First, suppose that one sub-block has $n-k+3t/2$ erasures and the other two have each $n-k-t$ erasures or less. From symmetry assume that the left sub-block has the largest number. Then decoding the center sub-block with a $[n,k+t]$ code gives $\vect{v}_{2,1},\ldots,\vect{v}_{2,4}$, and similarly the right block gives $\vect{v}_{3,1},\ldots,\vect{v}_{3,4}$. In addition, we get $\vect{v}_{1,1}$ as the reverse map of $GF$ in the center sub-block, and then $\vect{v}_{1,2}$ by canceling $\vect{v}_{3,3}$ and $\vect{v}_{3,4}$ from the reverse map of $GE$ in the center sub-block. Finally, we can find $\vect{v}_{1,3}+\vect{v}_{1,4}$ by canceling $\vect{v}_{2,2}$ from the reverse map of $GE$ in the right sub-block. Now we can decode the left sub-block with a $[n,k-3t/2]$ code having contributions from only $GI$ and $G4$. Second, suppose that two sub-blocks have each $n-k-t/2$ erasures and the third one has $n-k-t$ erasures or less. From symmetry assume that the right sub-block has the smaller number. Then decoding the right sub-block with a $[n,k+t]$ code gives $\vect{v}_{2,1}$, and $\vect{v}_{3,1}$ (among all $\vect{v}_{3,l}$). Now after canceling these two we can decode each of the left and center sub-blocks with a $[n,k+t/2]$ code and recover their erased symbols.

\section{Conclusion}
The constructions in this paper are the first known ones to offer read access in both sub-blocks and full blocks. The first construction has optimal $\Dst_1$-minimum-distance given $t$, but a limiting reduction in the $\Dst_2$-minimum-distance as $t$ grows. The second construction has slower growth of the $\Dst_1$-minimum-distance, but in return also slower reduction in $\Dst_2$, and altogether a higher maximal minimum distance. Future work is needed to extend these constructions, and find others with similar tradeoffs, to additional $m$ values. While the codes suffice to use finite fields with sizes as small as the sub-block size, further reducing the field size is an important future direction, e.g. through using BCH codes as the constituents.

\section{Acknowledgement}
This work was supported in part by the German-Israel Foundation and by the Israel Science Foundation.

\bibliography{myrefs1}
\end{document}